\documentclass{amsart}

\usepackage{graphicx}
\usepackage{stmaryrd}
\usepackage{amssymb}
\usepackage{amsmath}
\usepackage{amsthm}
\usepackage{dsfont}
\usepackage{hyperref}
\usepackage{amsfonts}
\usepackage{amsrefs}

\newtheorem{theorem}{Theorem}[section]
\newtheorem{lemma}[theorem]{Lemma}

\newtheorem{corollary}[theorem]{Corollary}
\theoremstyle{definition}
\numberwithin{equation}{section}

\newcommand{\ES}{E}
\newcommand{\ue}{e}
\newcommand{\ug}{g}
\newcommand{\ve}{\vec{e}}
\newcommand{\vg}{\vec{g}}
\newcommand{\ugamma}{{\gamma}}
\newcommand{\vgamma}{\vec{\gamma}}
\newcommand{\fE}{\mathcal{E}}
\newcommand{\dG}{\vec{\mathcal{G}}}
\newcommand{\Out}{\ensuremath{\text{Out}}}
\newcommand{\In}{\text{In}}
\newcommand{\KW}{T}
\newcommand{\Proj}{\text{Proj}}
\newcommand{\Id}{\mathrm{Id}}	
\newcommand{\Walk}{\omega}								
\newcommand{\WalkW}{\mathcal{W}}
\newcommand{\WalkU}{{\mathcal{U}}}
\newcommand{\WalkV}{\mathcal{V}}
\newcommand{\C}{\mathds{C}}
\newcommand{\R}{\mathds{R}}
\newcommand{\Z}{\mathds{Z}}
\newcommand{\LL}{\mathcal{L}}
\newcommand{\XX}{\mathcal{X}}
\newcommand{\CI}{S}
\newcommand{\Sf}{\varphi}
\newcommand{\FinG}{\mathcal{G}}
\newcommand{\FinH}{\mathcal{H}}
\newcommand{\InfG}{\Gamma}
\newcommand{\dx}{\vec{x}}

\newcommand{\dist}{d}

\begin{document}
\date{June 10, 2013}
\title[Fermionic observable and inverse Kac-Ward operator]
	{The fermionic observable in the Ising model and the inverse Kac-Ward operator}
\author{Marcin Lis}
\address{VU University\\
Department of Mathematics\\
De Boelelaan 1081a\\
1081\,HV Amsterdam\\
The Netherlands}
\email{m.lis\,@\,vu.nl}
\begin{abstract}
We show that the critical Kac-Ward operator on isoradial graphs acts in a certain sense as the operator of s-holomorphicity, 
and we identify the fermionic observable for the spin Ising model as the inverse of this operator.
This result is partially a consequence of a more general observation that the inverse Kac-Ward operator on any planar graph 
is given by what we call a fermionic generating function. Furthermore, using bounds obtained in \cite{Lis} for the spectral radius and 
operator norm of the Kac-Ward transition matrix, we provide a general picture of the non-backtracking walk representation of the critical 
and supercritical inverse Kac-Ward operators on isoradial graphs.
\end{abstract}


\maketitle

\section*{Introduction} 
The discrete fermionic observable for the FK-Ising model on the square lattice was introduced by Smirnov in~\cite{Smir2} (although, as 
mentioned in~\cite{ChelkSmir}, similar objects appeared in earlier works). He proved in~\cite{Smir1} that its scaling limit at 
criticality is given by the solution to a Riemann-Hilbert boundary value problem, and therefore is conformally covariant. A generalization 
of this result to Ising models defined on a large class of isoradial graphs was obtained by Chelkak and Smirnov in~\cite{ChelkSmir}, 
yielding also universality of the scaling limit. 

Since then, several different types of observables have been proposed for both the random cluster and classical spin Ising model.
They were used to prove conformal invariance of important quantities in these models. The scaling limit of the energy density of the critical 
spin Ising model on the square lattice was computed by Hongler and Smirnov~\cite{HongSmir}. Existence and conformal invariance of the scaling 
limits of the magnetization and multi-point spin correlations were established by Chelkak, Hongler and Izyurov~\cite{CHI}. The observable also 
proved useful in the off-critical regime and was employed by Beffara and Duminil-Copin~\cite{BefDum} to give a new proof of criticality of the 
self-dual point and to calculate the correlation length in the Ising model on the square lattice. In a more recent work of Hongler, Kyt\"{o}l\"{a} 
and Zahabi~\cite{HKZ}, the fermionic observables were identified as correlation functions of fermion operators in the transfer matrix formalism for 
the same model. One also has to mention the relation between the fermionic observable and the inverse Kasteleyn operator which was pointed out by 
Dub{\'e}dat~\cite{Dubedat}.

In this paper, we establish a direct connection between the fermionic observable for the spin Ising model and the inverse Kac-Ward operator. 
The method of Kac and Ward~\cite{KacWard} is a way of expressing the square of the partition function of the Ising model on a planar graph 
as the determinant of the Kac-Ward operator. It was proposed as a combinatorial alternative to the purely algebraic approach developed by Onsager 
and Kaufman~\cites{Onsager, Kaufman}. The Kac-Ward method and the Kac-Ward operator itself have recently become an object of revived interest. 
A thorough treatment of this approach with a focus on the combinatorics of configurations of loops was presented by Kager, Meester and the author in \cite{KLM}. 
The method was used there to rederive the critical point of the Ising model on the square lattice and to obtain new expressions for the 
free energy density and spin correlation functions in terms of signed loops and non-backtracking walks in the graph.
The same formulas for finite graphs were independently obtained by Helmuth in~\cite{Helm}, where the Kac-Ward method was put into a more 
general combinatorial context of heaps of pieces. Also in~\cite{Helm}, the spinor fermionic observable from~\cites{CI, CHI} was explicitly 
identified in terms of non-backtracking walks, though without addressing the issues of convergence of the expansions.  
In~\cite{Lis} the author, following the ideas contained in~\cite{KLM}, obtained bounds on the spectral radius and operator norm 
of the Kac-Ward transition matrix on a general graph and proved criticality of the self-dual Z-invariant Ising model introduced by Baxter~\cite{Baxter}.
Other relevant examples are the extension of the Kac-Ward method to graphs of higher genus introduced by Cimasoni~\cite{Cimasoni} and the 
computation of the critical temperature of the Ising model on doubly periodic planar graphs performed by Cimasoni and Duminil-Copin~\cite{CimDum}.

This paper consists of three sections. In Section~\ref{sec:genfunctions} we define the Kac-Ward operator on a general graph in the plane.
We then describe properties of the complex weights induced by this operator on the non-backtracking walks in the graph. In the end, we use loop expansions of
the even subgraph generating function from \cite{KLM} to express the inverse Kac-Ward operator on a finite graph
 in terms of a weighted sum over a certain family of subgraphs. 
We call the resulting formula the fermionic generating function since it bears a strong resemblance to the definitions of the spin fermionic 
observables from~\cites{ChelkSmir,HongSmir,HKZ}. In Section~\ref{sec:isoradial} we work on isoradial graphs. 
First, we consider the Kac-Ward operator corresponding to the 
critical Ising model, and we show that it can be thought of as the operator of s-holomorphicity. 
Subsequently, using bounds from~\cite{Lis}, we show that in finite volume the inverted critical Kac-Ward operator 
admits a representation in terms of non-backtracking walks, whereas a continuous inverse in infinite volume does not exist.
We also consider the supercritical inverse operators. They too are expressed in terms of walks (both on finite and infinite graphs),
and moreover, the associated Green's function decays exponentially fast with the distance between two edges. 
In particular, the supercritical operator on the full isoradial graph has a continuous inverse.

As a remark, we would like to point out that our observations seem to fit into
 a more general picture of two-dimensional discrete physical models satisfying the following three conditions:
\begin{itemize}
\item[\emph{(i)}] the partition function of the model is equal to the square root of the determinant of some operator,
\item[\emph{(ii)}] an important observable in the model is given by the inverse of this operator,
\item[\emph{(iii)}] the critical values of parameters of the model coincide with the values of parameters which make this operator into some (massless) discrete differential operator.
\end{itemize}
Our results show that the Ising model on isoradial graphs satisfies this classification with the distinguished operator being 
the Kac-Ward operator, the observable being the fermionic observable, and the discrete differential operator being the s-holomorphic operator.
 Another example is the discrete Gaussian free field, where the partition function is equal 
to the square root of the determinant of the discrete Laplacian, and the two-point spin correlation functions are given by the inverse of the 
Laplacian. Moreover, the general picture of the non-backtracking walk representation of the inverse Kac-Ward operator presented in Section~\ref{sec:isoradial} 
matches the one of the random walk representation of the inverse Laplacian \cite{BFS}.
Also the dimer model~\cite{Kasteleyn}, which is known to be closely related to the Ising model, fits this pattern. The square of 
the partition sum of this model is equal to the absolute value of the determinant of the Kasteleyn operator, which acts as the discrete Dirac 
operator (see e.g.~\cite{Ken1}). In addition, the observable of main interest in the work of Kenyon~\cite{Ken1} is the coupling function defined as the inverse of the Kasteleyn operator.

\section{The Kac-Ward operator and graph generating functions}
\label{sec:genfunctions}
We will consider graphs embedded in the complex plane. To simplify our notation, we will identify graphs with their edge sets. By an (\emph{undirected})
 \emph{edge} we mean an unordered pair of distinct complex numbers $\{z ,w\}$, which we will identify with the closed straight line segment in the complex plane 
connecting $z$ and $w$. We say that $z$ and $w$ are the \emph{endpoints} of $\{z,w\}$. A collection of edges $\FinG$ is called a \emph{graph} if any two edges 
in $\FinG$ share at most one point, which is either an endpoint of both of them, or is not an endpoint of any edge in $\FinG$.
In the latter case we say that the two edges \emph{cross} each other. A complex number $z$ is called a \emph{vertex} in $\FinG$ if it is an endpoint of some edge in $\FinG$.

Although a graph $\FinG$ is undirected, we will mainly work with the \emph{directed edges} of $\FinG$, i.e.\ the ordered pairs of complex numbers $(z,w)$ for which $\{z,w\} \in \FinG$. 
We write $\dG$ for the set of all directed edges of $\FinG$. We will always denote a directed edge by a letter with an arrow over it, whereas the undirected counterpart will be obtained 
by dropping the arrow from the notation, i.e.\ if $\ve=(z,w)$, then $\ue=\{z,w\}$. For a directed edge $\ve=(z,w)$, we define its \emph{tail} $t(\ve)=z$ and its \emph{head} $h(\ve) =w$, 
and we write $-\ve$ for the inverted edge $(w,z)$.

\subsection{The Kac-Ward operator} 
From now on, we assume that $\FinG$ is a fixed (possibly infinite) graph with finite maximal degree, and $x=(x_{\ue})_{\ue\in \FinG}$ is a system of positive edge 
weights satisfying $\| x \|_{\infty}=\sup_{\ue \in \FinG} x_{\ue} < \infty$.
For two directed edges $\ve$ and $\vg$, let
\begin{align} \label{eq:defangle}
\angle(\ve,\vg)= \text{Arg}\Big(\frac{h(\vg)-t(\vg)}{h(\ve)-t(\ve)}\Big) \in (-\pi,\pi]
\end{align}
be the \emph{turning angle} from  $\ve$ to $\vg$. 
The \emph{(Kac-Ward) transition matrix} is given by 
\begin{align} \label{eq:transitionmatrix}
\Lambda_{\ve,\vg} = \begin{cases}
		x_{\ue} e^{\frac{i}{2}\angle(\ve,\vg)}
		& \text{if $h(\ve) = t(\vg)$ and $\vg \neq -\ve$}; \\
		0 & \text{otherwise}.
	\end{cases}
\end{align}
The \emph{Kac-Ward operator} is an automorphism of the complex vector space~$\C^{\dG}$ defined via matrix multiplication by the matrix
\begin{align*}
\KW = \Id -\Lambda,
\end{align*}
where $\Id$ is the identity matrix. 
Note that this is well defined since $T$ has at most $\Delta$ nonzero entries in each row, where $\Delta$ is the maximal 
degree of~$\FinG$. Moreover, since the weight system is bounded, $\KW$ is continuous (bounded) when treated as an operator on $\ell^2(\dG)$.
It is known, that if $\FinG$ is finite and no two edges in $\FinG$ cross each other, 
then the determinant of $T$ is strictly larger than one (see~Theorem \ref{thm:main}).
In particular, in this case $T$ is an isomorphism.

\subsection{Weighted non-backtracking walks} \label{sec:walks}
A \emph{(non-backtracking) walk} $\Walk$ of \emph{length} $|\Walk|=n$ in $\FinG$ is a sequence of directed edges $\Walk=(\ve_0,\ve_1,\ldots,\ve_{n}) \in \dG^{n+1}$,
 such that $h(\ve_i)=t(\ve_{i+1})$ and $\ve_{i+1} \neq -\ve_{i}$ for $i=0,\ldots,n-1$. Note that $|\Walk|$ counts the number of steps $\Walk$ makes between edges, 
rather than the number of edges it visits.  A walk $\Walk$ is \emph{closed} if $\ve_0=\ve_n$ and $|\Walk|>0$. We say that $\Walk$ \emph{goes through} a directed
 edge $\ve$ (undirected edge~$\ue$) if $\ve_i = \ve$ ($\ue_i = \ue$), for some $i\in \{ 0,\ldots,n-1 \}$. Note that $\Walk$ does not
necessarily go through~$\ve_n$, and in particular, walks of length zero do not go through any edge. 
By $\ES(\Walk)$ we denote the \emph{edge set} of $\Walk$, i.e.\ the set of undirected edges that the walk $\Walk$ goes through. A walk is called a \emph{path}
 if it goes through every undirected edge at most once. By $-\Walk$ we mean the reversed walk $(-\ve_n,-\ve_{n-1},\ldots,-\ve_{0})$.

The \emph{(signed) weight} of a walk $\Walk=(\ve_0,\ve_1,\ldots,\ve_{n})$ is given by
\begin{align} \label{eq:defweight}
w(\Walk) = e^{\frac{i}{2} \alpha(\Walk)}  \prod_{i=0}^{n-1} x_{\ue_i},\qquad \text{where} \qquad \alpha(\Walk) = \sum_{i=0}^{n-1} \angle(\ve_{i},\ve_{i+1}) 
\end{align}
 is the total turning angle of $\Walk$. Note that with this definition of the signed weight, the last edge of $\Walk$ is not counted in terms 
of edge weights but does contribute to the total winding angle of $\Walk$. In particular, $w(\Walk)$ is a monomial in the 
variables $x_{\ue}$, $\ue \in \ES(\Walk)$. If $|\Walk|=0$, then we put 
$\alpha(\Walk)=0$ and $w(\Walk)=1$. The fundamental feature of the signed weight is that it factorizes over the steps that a path makes, 
 where the step weight is given by the transition matrix \eqref{eq:transitionmatrix}, i.e.
\begin{align} \label{eq:stepfact}
w(\Walk) = \prod_{i=0}^{n-1} \Lambda_{\ve_{i},\ve_{i+1}}.
\end{align} 

Given two directed edges $\ve$ and $\vg$, we write $\WalkW(\ve,\vg)$ for the collection of all walks in $\FinG$ which start at $\ve$ and end at $\vg$. 
Since the complex argument satisfies the logarithmic identity $\text{Arg}(z/w) =\text{Arg}(z)-\text{Arg}(w) \text{ (mod $2\pi $)}$, we conclude that
\begin{align}
		w(\Walk) \in e^{\frac{i}{2}\angle(\ve,\vg)} \R \qquad  \text{for } \Walk \in  \WalkW(\ve,\vg).
\end{align}
On the other hand, since walks are non-backtracking and $\text{Arg}(1/z)=-\text{Arg}(z)$ for $z \notin (-\infty,0]$, it follows that $\alpha(\Walk)= -\alpha(-\Walk)$.
Combining these two facts, we obtain that
\begin{align} \label{eq:inverted}
w(\Walk) = \begin{cases}
			- w(-\Walk) & \text{ if } \Walk \in  \WalkW(\ve,-\ve); \\
		 \ \ w(-\Walk) & \text{ if } \Walk \in  \WalkW(\ve,\ve).
	\end{cases}
\end{align}

The first identity in \eqref{eq:inverted} implies cancellations of weights of walks which go through certain edges in both directions.
The most basic consequence of this property is the following lemma:
\begin{lemma} \label{lem:zerosuminverse} For any $\ve \in \dG$,
\begin{align*}
		\sum_{\Walk \in \WalkW(\ve,-\ve)} w(\Walk) =0.
\end{align*}
\end{lemma}
\begin{proof}
If $\WalkW(\ve,-\ve)$ is empty then the above statement is trivially true.
Otherwise, if $\Walk \in \WalkW(\ve,-\ve)$, then $-\Walk \in \WalkW(\ve,-\ve)$, $w(\Walk)=-w(-\Walk)$, and $-(-\Walk)=\Walk$. Hence, we have cancellation of all terms in the series.
\end{proof}
This observation, and other which naturally follow from property \eqref{eq:inverted} (see Lemma \ref{lem:sumsloops} 
and \ref{lem:walkfact} in Section \ref{sec:proofs}) will be important in the computation of the inverse of the Kac-Ward operator. 

Note that the above sum is in general an infinite power series in the variables~$x_{\ue}$. In order to be rigorous when dealing with power series, 
we will always assume, unless stated otherwise, that $\|x\|_{\infty}$ is sufficiently small for the series to be absolutely convergent. 
In all of the cases, it will be enough to take $\|x\|_{\infty} < 1/(\Delta-1)$.

\subsection{Crossings in graphs and walks}
As mentioned before, edges of a graph can cross each other. For a finite graph $\FinH$, 
let $C(\FinH)$ be the number of \emph{edge crossings} in $\FinH$, i.e.\
the number of unordered pairs of edges in $\FinH$ which cross each other.

A similar notion of a crossing can be assigned to closed walks. 
One can think of a closed walk $\Walk=(\ve_0,\ldots,\ve_n)$ as a closed continuous curve in the complex plane with the time parametrization given by
\begin{align*}
\Walk(s) = t(\ve_{\lfloor s \rfloor})+ (s -\lfloor s \rfloor)\big(h(\ve_{\lfloor s \rfloor})-t(\ve_{\lfloor s \rfloor})\big) \qquad \text{for} \ s \in \R / n\Z.
\end{align*}
With this definition, we say that $(s_1,s_2) \in (\R / |\Walk_1|\Z) \times (\R / |\Walk_2|\Z)$ is a \emph{ crossing} at $z$ between two closed walks 
$\Walk_1$ and $\Walk_2$, if $z = \Walk_1(s_1) = \Walk_2(s_2)$, and for any small 
open neighborhood $U_1 \times U_2$ of $(s_1,s_2)$,  there is a small ball $B$ around $z$, 
such that $\Walk_1(U_1)$ intersects both connected components of $B \setminus \Walk_2(U_2)$. 
Note that in our setting, the only possible crossings between closed walks can occur at the vertices 
of $\FinG$, or at the points of crossings of edges of $\FinG$. 
We write $C(\Walk_1,\Walk_2)$ for the number of crossings between $\Walk_1$ and~$\Walk_2$, 
and $C(\Walk) = C(\Walk,\Walk)/2$ for the number of \emph{self-crossings} of $\Walk$.

We say that two walks are \emph{edge-disjoint} if their edge sets are disjoint.
For topological reasons, if $\Walk_1$ and $\Walk_2$ are edge-disjoint closed paths, then $C(\Walk_1,\Walk_2)$ is even. Moreover, there is an intrinsic connection 
between the total turning angle of a closed path and the number of its self-crossings:
\begin{theorem}[Whitney, \cite{Whitney}]\label{thm:Whitney} For any closed path $\Walk$,
\begin{align*}
-e^{\frac{i}{2} \alpha(\Walk)} = (-1)^{C(\Walk)}.
\end{align*}
\end{theorem}

\subsection{Generating functions of even subgraphs}
We call a graph $\FinH$ \emph{even} if all its vertices have even degree. Equivalently, a finite graph $\FinH$ is even if and only if it is 
a union of edge sets of some collection of edge-disjoint closed paths. The \emph{generating function of even subgraphs} 
of a finite graph $\FinG$, as defined in~\cite{KLM}, is given by
\begin{align} \label{eq:defZ}
Z = \mathop{\sum_{\FinH \subset \FinG}}_{\FinH \text{ even }} (-1)^{C(\FinH)}\prod_{\ue \in \FinH} x_{\ue}.
\end{align}
The empty set is also even, and we assume that its contribution to $Z$ equals one. If $\FinG$ has no edge crossings, then $Z$ counts all even subgraphs
 with positive sign, and in particular is bigger than one. If additionally $\|x\|_{\infty} \leq 1$, then by the \emph{high-temperature expansion}, $Z$ is proportional 
to the partition function of the Ising model on $\FinG$ with free boundary conditions and with appropriate coupling constants (see e.g. \cite{KLM}).
However, it will be crucial for the computation of the inverse Kac-Ward operator to allow graphs with crossings, and 
therefore we will need the following result:
\begin{theorem}{\cite{KLM}*{Theorem 1.9}} \label{thm:main}
\begin{align*}
Z = \exp\Big( -\sum_{\Walk \text{ closed} } \frac{w(\Walk )}{2|\Walk|}\Big) = \sqrt{{\det}\KW},
\end{align*}
where the sum is over all closed walks in $\FinG$. 
\end{theorem}
The first equality of this theorem yields a direct connection between the generating function of even subgraphs and the signed non-backtracking walks.
Note that the notation used here differs from the one in~\cite{KLM}. In particular, the signed weight from~\cite{KLM} is minus the signed weight defined
 in~\eqref{eq:defweight}. Also the corresponding exponential formula in~\cite{KLM} is written in terms of \emph{loops}, 
i.e.\ equivalence classes of closed walks defining the 
same, up to a time parametrization, closed curve in the plane. The signed weight of a loop is then defined as the sum of signed weights of closed walks in the
equivalence class.

\subsection{The inverse Kac-Ward operator} \label{sec:InverseKW}
In this section we assume that $\FinG$ is finite and without edge crossings.
If one wants to compute the inverse of the Kac-Ward operator, one can use the power series formula:
\begin{align} \label{eq:powerseries}
 \KW^{-1}_{\ve,\vg} =(\Id - \Lambda )^{-1}_{\ve,\vg} = \sum_{n=0}^{\infty} \Lambda_{\ve,\vg}^n =  \sum_{\Walk \in \WalkW(\ve,\vg)} w(\Walk), 
\end{align}
which is valid for $\|x\|_{\infty}$ small enough. The last sum is over all non-backtracking walks since the transition matrix $\Lambda$ assigns zero weight to steps between $\ve$ and $-\ve$.
It turns out that this sum can be expressed in terms of a generating function of certain subgraphs of $\FinG$ (or rather its particular modification).

To this end, let $m(\ve)=(t(\ve)+h(\ve))/2$ be the midpoint of $\ve$.
Given $\ve,\vg \in \dG$, we define a modified graph 
\[
\FinG_{\ve,\vg} = \big( \FinG \setminus \{ \ue,\ug\} \big) \cup  \big \{ \{ m(\ve), h(\ve) \}, \{t(\vg), m(\vg) \} \}
\] 
which, instead of $\ue$ and $\ug$, contains appropriate \emph{half-edges}.
The weight of $\{ m(\ve), h(\ve) \}$ is set to be $x_{\ue}$, and in the case when $\vg \neq -\ve$, the weight of $\{t(\vg), m(\vg) \}$ is one. 
We write $\fE(\ve,\vg)$ for the collection of subgraphs $\FinH \subset \FinG_{\ve,\vg}$ containing the half-edges $ \{ m(\ve), h(\ve) \}$ and $\{t(\vg), m(\vg)\} $,
 and such that all vertices of $\FinG$ have even degree in $\FinH$ (see Figure \ref{fig1}). Note that we do not require that $m(\ve)$ and $m(\vg)$ have even degree.
It follows that $\fE(\ve,-\ve)$ is empty, since there is no graph which has exactly one vertex with odd degree. Also note that $\fE(\ve,\ve)$ is 
in bijective correspondence with the set of even subgraphs of~$\FinG$ containing~$\ue$.

Suppose that $\vg \neq -\ve$ and take $\FinH \in \fE(\ve,\vg)$. It follows that there is a path in $\FinH$ which starts
at $(m(\ve), h(\ve)) $ and ends at $(t(\vg), m(\vg))$.
Let $\Walk_{\FinH}$ be the left-most such path, i.e.\ the path which always makes a step to the left-most edge which has not yet been visited 
in any direction. Note that $\FinH$ splits into $\Walk_{\FinH}$ and an even subgraph of $\FinG$ (see Figure \ref{fig1}).
Since $\FinH$ also belongs to $\fE(-\vg,-\ve)$ this notation may be ambiguous (the reversed left-most path becomes the right-most path), but we will always 
use it in the context of fixed edges $\ve$ and $\vg$.

For $\ve,\vg \in \dG$, we define the \emph{fermionic generating function} by 
\begin{align} \label{eq:deffermion}
F_{\ve,\vg} = \delta_{\ve,\vg}+\frac{1}{Z} \sum_{\FinH \in \fE (\ve,\vg)} e^{-\frac{i}{2} \alpha(\Walk_{\FinH})}\prod_{h \in \FinH} x_{h},
\end{align}
where $\delta$ is the Kronecker delta.
For $\vg = -\ve$, the above sum is empty and we take it to be zero.

Note the resemblance between this definition and the definitions of
fermionic observables from \cites{ChelkSmir,HongSmir,HKZ}. The important difference is that the fermionic generating function is a function of two 
directed edges and the fermionic observable from the literature can be seen as a function of one directed and one undirected edge. Indeed, for regular lattices 
(the square, triangular and hexagonal lattice), the fermionic observable is, up to a complex multiplicative constant, the symmetrization in the 
variable $\vg$ of the fermionic generating function, i.e.\ the sum over the two opposite orientations of the undirected edge $\ug$. For general isoradial 
graphs it becomes a weighted symmetrization, where the weight depends on the local geometry of the graph (see Section~\ref{sec:isoradial}).

Recall that we assume that $\FinG$ is finite and does not have any edge crossings.
We can now state the main theorem of this section:

\begin{theorem} \label{thm:fermionwalks}
For any $\ve,\vg \in \dG$,
\begin{align*}
		\overline{F_{\ve,\vg}} = \sum_{\Walk \in \WalkW(\ve,\vg)} w(\Walk).
\end{align*}
\end{theorem}
For the proof of this result, see Section~\ref{sec:prooffermionwalks}.
As a direct corollary, we get that $\overline{F}=\big(\overline{F_{\ve,\vg}}\big)_{\ve,\vg \in \dG}$ is the inverse Kac-Ward operator:

\begin{figure}
		\begin{center}
			\includegraphics{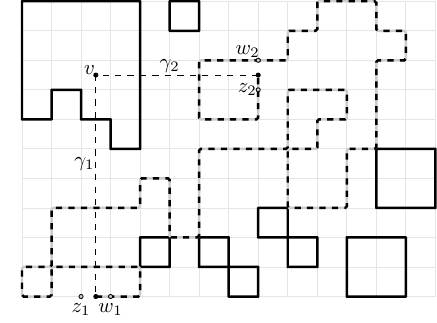}
		\end{center}
		\caption{In this case, $\FinG$ is a rectangular piece of the square lattice. A graph $\FinH \in \fE (\ve,\vg)$ is drawn in bold lines, 
where $\ve=(z_1,w_1)$ and $\vg=(z_2,w_2)$. The graph splits into two parts: the path $\Walk_{\FinH}$ represented by the dashed lines, and an even 
subgraph of $\FinG$. Adding two edges $\gamma_1$ and $\gamma_2$ makes $\Walk_{\FinH}$ into a closed path with three self-crossings, and $\FinH$ 
into an even subgraph of $\FinG_{\ve,\vg} \cup \{ \gamma_1,\gamma_2\}$ with five edge crossings (see the proof of Theorem \ref{thm:fermionwalks}).}
		\label{fig1}
\end{figure}

\begin{corollary} \label{cor:Tinverse} The inverse Kac-Ward operator on a finite graph $\FinG$ with no edge crossings 
is the complex conjugate of the fermionic generating function, i.e.\
\begin{align*}
	\KW^{-1} = \overline{F}.
\end{align*}
\end{corollary}
\begin{proof} Fix $\ve,\vg \in \dG$, and a weight system $x$. 
Consider the rescaled system of weights $tx$, where $t$ is a positive real number. 
Since $\FinG$ has no edge crossings, $Z$ is never zero by \eqref{eq:defZ} and it follows from Theorem \ref{thm:main} that $\det \KW$ is also never zero.
Hence, $\overline{F_{\ve,\vg}}$ and $\KW^{-1}_{\ve,\vg}$, treated as functions of the scaling factor~$t$, are analytic on $(0,\infty)$. By 
uniqueness of the analytic continuation, it is enough to prove the desired equality for $t$ small, and this follows from Theorem \ref{thm:fermionwalks}
and the power series expansion \eqref{eq:powerseries}.
\end{proof}
Note that the fermionic generating function was defined only for finite graphs. Theorem \ref{thm:fermionwalks} and Corollary \ref{cor:Tinverse} give
two interpretations of $F$ which do not require finiteness of the underlying graph. We will discuss this issue in Section \ref{sec:isoradial}.

\section{The Kac-Ward operator on isoradial graphs} 
\label{sec:isoradial}
In this section we assume that $\FinG$ is a subgraph of an infinite isoradial graph~$\InfG$.  This means that all faces of $\InfG$ 
can be inscribed into circles with a common radius and the circumcenters lie within the corresponding faces. Equivalently, the dual graph $\InfG^*$ can
be embedded in such a way that all pairs of mutually dual edges $\ue$ and $\ue^*$ form diagonals of rhombi.
For each edge~$\ue$, let~$\theta_{\ue}$ be the undirected 
angle between $\ue$ and any side of the rhombus associated to~$\ue$ (see Figure \ref{fig2}). We will consider a family of weight 
systems given by
\begin{align} \label{eq:isingweights}
x_{\ue}(\beta) = \tanh \beta J_{\ue}, \qquad \text{where} \qquad \tanh J_{\ue}=\tan (\theta_{\ue}/2),
\end{align}
and where $\beta \in (0,1]$ is the \emph{inverse temperature}.
These weights come from the high-temperature expansion of the Ising model and the numbers 
$J_{\ue}$ are called the \emph{coupling constants} (see e.g.~\cite{KLM}).
In the case when $\beta=1$ we will talk about the \emph{critical} weight system, and for $\beta \in (0,1)$ the
weights will be called \emph{supercritical}. The critical case corresponds to the self-dual Z-invariant Ising model which was introduced by 
Baxter~\cite{Baxter} and which has been extensively studied in the mathematical literature.
Chelkak and Smirnov proved in \cite{ChelkSmir} that the critical fermionic observable has a universal, conformally invariant scaling limit. 
Boutillier and de Tili{\`e}re~\cites{BoutTil1, BoutTil2} analysed the model using the dimer representation, and the author \cite{Lis} proved that, 
after introducing the inverse temperature parameter, the model has a phase transition at  $\beta=1$ and nowhere else.

\subsection{The critical Kac-Ward operator and s-holomorphicity}
In this section we assume that the weight system is critical.
The notion of \emph{s-holomorphicity} (s stands for \emph{strong} or \emph{spin}) was introduced in \cite{Smir1} in the setting of the square lattice, 
and was later generalized in~\cite{ChelkSmir} to fit the context of general isoradial graphs.
Our definition of s-holomorphicity will be equivalent to that in \cite{ChelkSmir}, up to multiplication of the function by some globally fixed complex constant.

\begin{figure}
		\begin{center}
			\includegraphics{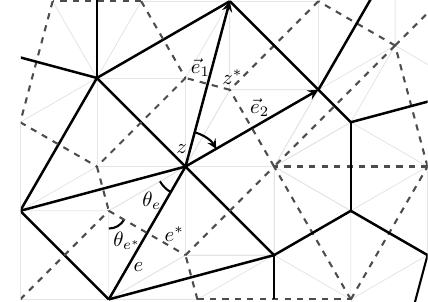}
		\end{center}
		\caption{A local picture of an isoradial graph and its dual. The directed arc marks the turning angle $\angle(\ve_1,\ve_2)$.}
		\label{fig2}
\end{figure}

Consider a vertex $z$ in $\InfG$ and let $z^*$ be a vertex in $\InfG^*$ corresponding to one of the faces of 
$\InfG$ incident to $z$. By $\ue_1$ and $\ue_2$ we denote the two edges lying on the boundary of this face and having $z$ as an endpoint (see Figure \ref{fig2}). 
We say that a complex function $f$ defined on the edges of $\InfG$ is \emph{s-holomorphic} at~$z$ if 
for all such dual vertices $z^*$ and the corresponding edges $\ue_1$ and $\ue_2$, 
\begin{align*}
\Proj(f(\ue_1);  (z-z^*)^{-\frac{1}{2}}\R) = \Proj(f(\ue_2);  (z-z^*)^{-\frac{1}{2}}\R),
\end{align*}
where $\Proj(w;  \ell)$ is the orthogonal projection of the complex number $w$ onto the line $\ell$. 
Note that the choice of the square root is immaterial in the definition above.
The property of being s-holomorphic is a real linear property, i.e.\ addition of two functions and multiplication of a function by a real number preserves s-holomorphicity. 
It is also a stronger property than the usual discrete holomorphicity: if a function is s-holomorphic at $z$, then 
the same function considered as a function on the dual edges is discrete holomorphic at $z$, i.e.\ the
discrete contour integral around the face corresponding to $z$ vanishes.
On the other hand, each discrete holomorphic function is, up to an additive constant, uniquely represented as a sum of two s-holomorphic functions,
where one of them is multiplied by~$i$. For proofs of these facts and other properties of s-holomorphic functions, see~\cite{ChelkSmir}.

The Kac-Ward operator was defined in Section \ref{sec:genfunctions} as an automorphism of the complex vector space $\C^{\dG}$ but it can also be seen as
an operator acting on a smaller real vector space. To be precise, to each directed edge $\ve$ we associate a line $\ell_{\ve}$ in the complex plane defined by
\begin{align*}
\ell_{\ve} = e^{-\frac{i}{2} \angle(\ve)} \R,\qquad \text{where} \quad \angle(\ve)= \text{Arg} (h(\ve) - t(\ve)).
\end{align*}
As before, we use the principal value of the complex argument.
Note that $\ell_{\ve}$ and $\ell_{-\ve}$ are orthogonal and they can be thought of as ``a local coordinate system at $\ue$''.
We will consider the direct product of the lines treated as one-dimensional real vector spaces, i.e.\ we put
\begin{align*}
	\LL = \prod_{\ve \in \dG} \ell_{\ve}.
\end{align*}
By the logarithmic property of the complex argument, $T_{\ve,\vg}$ defines by multiplication a linear map from $\ell_{\vg}$ to $ \ell_{\ve}$.
This means that the Kac-Ward operator can be seen as an automorphism of $\LL$.
We define $\XX$ to be $\C^{\FinG}$ treated as a real vector space and we consider an isomorphism between $\XX$ and $\LL$ given by
\begin{align*}
\CI f(\ve) = \sin(\theta_{\ue}/2) \Proj (f(\ue) ; \ell_{\ve})\qquad \text{for } f \in \XX.
\end{align*}

If $\InfG$ is a regular lattice (the square, triangular or hexagonal lattice), then all the angles $\theta_{\ue}$ are equal and $S$ is proportional to the projection operator
which gives ``local coordinates'' at each edge. Note that the inverse of $S$ is just a locally rescaled symmetrization operator, i.e.
\begin{align*}
S^{-1}\varphi(\ue) = \big(\sin (\theta_{\ue}/2)\big)^{-1}\big(\varphi(\ve)+\varphi(-\ve)\big)\qquad \text{for } \varphi \in \LL.
\end{align*}

We say that $z$ is an interior vertex of $\FinG$ if the degrees of $z$ in $\FinG$ and~$\InfG$ are the same.
The set of edges emanating from a vertex $z$ is denoted by $\Out(z)=\{\ve \in \dG: t(\ve)=z \}$, and $\In(z)=\{\ve \in \dG: h(\ve)=z \} = - \Out(z)$ 
are the edges pointing at $z$. The next result expresses the fact that the critical Kac-Ward operator (composed with $S$) can be seen as the operator of s-holomorphicity: 
\begin{theorem} \label{thm:KWshol}
Let $\KW$ be the critical Kac-Ward operator. A function $f\in \XX$ is s-holomorphic at an interior vertex~$z$ if and only if 
\begin{align*}
\KW \CI f(\ve) =0\qquad \text{for all } \ve \in \mathrm{In}(z).
\end{align*}
\end{theorem}
The proof of this theorem is given in Section \ref{sec:proofKWshol}. 

Consider the case where $\FinG$ is the full $\InfG$ and take $f$ to be equal to one everywhere. Of course, $f$ is s-holomorphic at all vertices of $\InfG$.
It follows from the theorem above that $\KW \CI f$ is equal to zero everywhere and hence the critical Kac-Ward operator for the full isoradial 
graph has a nontrivial kernel. In particular, it is not invertible on $\LL$, and therefore also on $\C^{\dG}$.

Let us go back to the case where $\FinG$ is a finite subgraph of $\InfG$. From Section~\ref{sec:InverseKW}, we know that the inverse Kac-Ward operator
exists for all weight systems on~$\FinG$. As a consequence of Theorem \ref{thm:KWshol},
 we can construct s-holomorphic functions by applying the inverse of $\KW S$ to functions which are zero almost everywhere. 
To this end, we define the standard basis of $\LL$ to be the set of functions $\{ i_{\ve} \}_{\ve \in \dG}$, where
$i_{\ve}(\vg)= e^{-\frac{i}{2} \angle(\ve) }\delta_{\ve,\vg}$.
It follows that $f_{\ve}=(\KW S)^{-1}i_{-\ve}$ is s-holomorphic at all interior vertices of $\FinG$ which are not $t(\ve)$, and is not s-holomorphic at~$t(\ve)$. We also have that
\begin{align*}
 f_{\ve}(\ug)& = S^{-1}\KW^{-1}i_{-\ve}(\ug) \\ &
 \sim \big(\sin(\theta_{\ug}/2)\big)^{-1} \big(\KW^{-1}_{\vg,-\ve}+\KW^{-1}_{-\vg,-\ve}\big)  \\
& \sim \big(\cos(\theta_{\ug}/2)\big)^{-1} \big(F_{\ve,\vg}+F_{\ve,-\vg}\big),
\end{align*}
where~$\sim$ means equality up to a multiplicative constant depending only on~$\ve$.
We used here Corollary \ref{cor:Tinverse}, the fact that $x_{\ue}T^{-1}_{\vg,\ve}= x_{\ug}\overline{T^{-1}_{-\ve,-\vg}}$, 
and the definition of the critical weight system. As mentioned before, the cosine term vanishes from this expression if $\InfG$ is a regular
lattice. Recalling the definition of $F$, one can see that $f_{\ve}$ is proportional to the critical fermionic observable used in 
\cites{ChelkSmir,HongSmir,HKZ}.

\subsection{The non-backtracking walk representation} 
In this section we provide a representation of the inverse Kac-Ward operator in terms of non-backtracking walks.
Note that we already used this idea in~\eqref{eq:powerseries} but only for weights which were sufficiently small
in the supremum norm. It turns out that the walk expansions on isoradial graphs 
are valid for both supercritical and critical weight systems, though their behavior is different in each of these cases.

We will use tools from \cite{Lis} and hence we need a regularity condition on~$\InfG$, i.e.\
 we will assume that there exist constants~$k$ and~$K$ 
such that 
\begin{align} \label{eq:isoradialcond}
0 < k \leq \theta_{\ue} \leq K < \pi \qquad \text{for all} \ e \in E(\InfG).
\end{align}
Geometrically, this means that the area of the underlying rhombi is uniformly bounded away from zero, or in other words,
the rhombi do not get arbitrarily thin.

For $\ve,\vg \in \dG$, we define $\WalkW_n(\ve,\vg) \subset \WalkW(\ve,\vg)$ to be the subcollection of all walks of length $n$, 
and let $\dist(\ve,\vg)$ be the distance between $\ve$ and $\vg$, i.e.\ the length of a shortest walk in $\WalkW(\ve,\vg)$.
All operators in the following statement are treated as operators on the Hilbert space $\ell^2(\dG)$. 

\begin{theorem} 
\label{thm:walkrepresentation}
If the weights are supercritical and $\FinG$ is a subgraph of $\InfG$, or the weights are critical and $\FinG$ is a finite subgraph of $\InfG$, 
then the inverse Kac-Ward operator is continuous and is given by the matrix
\begin{align*}
\KW^{-1}_{\ve,\vg} = \sum_{n=\dist(\ve,\vg)}^{\infty} \sum_{\Walk \in \WalkW_n(\ve,\vg)} w(\Walk).
\end{align*}

Moreover, in the supercritical case, there exist constants $C$ and $\epsilon <1$ such that
\begin{align*}
 \quad \Big| \sum_{\Walk \in \WalkW_n(\ve,\vg)} w(\Walk) \Big| \leq C \epsilon^n \qquad \text{for all} \ \ve, \vg \ \text{and} \ n.
\end{align*}
Furthermore, $C$ and $\epsilon$ depend only on $\beta$ and on the isoradial graph~$\InfG$, and 
do not depend on the particular choice of $\FinG$.

Finally, if $\FinG$ is the full $\InfG$, then the critical Kac-Ward operator does not have a continuous inverse.
\end{theorem}
Section~\ref{sec:proowalkrepresentation} is devoted to the proof of this result.
Note that this theorem and Corollary~\ref{cor:Tinverse} provide a natural definition 
of the supercritical fermionic generating function on infinite isoradial graphs. Furthermore, 
the critical fermionic observable on finite graphs also admits a representation in terms of non-backtracking walks. 

We would like to mention that one could also consider Kac-Ward operators with \emph{subcritical} weights on the dual graph $\Gamma^*$, i.e.\ weights given by $x_{\ue^*} = \exp(-2\beta J_{\ue})$, where
$\tanh J_{\ue}=\tan (\theta_{\ue}/2)$ and $\beta >1$. Since~$\Gamma^*$ is also isoradial, the corresponding analysis would be 
similar due to the Kramers-Wannier duality of the planar Ising model (see \cites{KLM,Lis}).

As mentioned in the introduction, the picture that Theorem \ref{thm:walkrepresentation} presents
matches the one of the random walk representation of the inverse Laplacian~\cite{BFS} on the square lattice.
Indeed, the inverse of the Laplacian in finite volume is given by the random walk Green's function. Off criticality, i.e.\ when the Laplacian is massive,
the Green's function decays exponentially fast with the distance between two vertices. As a result, the inverse of the massive
operator in the whole plane exists and is continuous. On the other hand, the full-plane massless Laplacian does not have a bounded inverse.

The crucial difference between these two representations seems to be the fact that the weights of walks induced by the Laplacian are positive, 
and therefore yield a measure, whereas the Kac-Ward weights for the non-backtracking walks are complex-valued. 
In particular, in Theorem~\ref{thm:walkrepresentation},
we have to group the walks by length. Otherwise, the series may diverge. On the other hand, this is not an issue
in the random walk representation.

\section{Proofs of main results} \label{sec:proofs}
\subsection{Proof of Theorem \ref{thm:fermionwalks}}
\label{sec:prooffermionwalks}
\subsubsection*{Cancellations of signed weights}

We already stated Lemma \ref{lem:zerosuminverse} as the simplest manifestation of the cancellations of signed weights of the non-backtracking walks.
For the proof of Theorem \ref{thm:fermionwalks}, we will also need two slightly more difficult consequences of property \eqref{eq:inverted}.
To this end, for $\ve,\vg \in \dG$, let $\WalkV(\ve,\vg) \subset \WalkW(\ve,\vg) $ be the collection of walks which go through $\ue$ 
exactly once, and if $\ue \neq \ug$, do not go through $\ug$ (recall from Section~\ref{sec:walks} what is meant for a walk to go through an edge). Note that $(\ve) \notin \WalkV(\ve,\ve)$.
Also, let $\WalkU(\ve,\vg) \subset \WalkW(\ve,\vg)$ be the collection of walks which do not go through $-\ve$ and $-\vg$. 
Note that $\WalkU(\ve,-\ve) =\emptyset$ and $(\ve) \in \WalkU(\ve,\ve)$. When necessary, we will denote the dependence of these 
collections on the underlying graph $\FinG$ in the subscripts, e.g. we will write $ \WalkW_{\FinG}(\ve,\vg)$. 

The first property says that the closed walks, which go through their starting edge in both directions, do not contribute
to the total sum of weights.
\begin{lemma} \label{lem:sumsloops}  For any $\ve\in \dG$,
\[
	\sum_{\Walk \in \WalkW(\ve,\ve)} w(\Walk)=\sum_{\Walk \in \WalkU(\ve,\ve)} w(\Walk) = \Big(1-\sum_{\Walk \in \WalkV(\ve,\ve)} w(\Walk)\Big)^{-1}.
\]
\end{lemma}
\begin{proof}
If $\mathcal{A}=\WalkW(\ve,\ve) \setminus \WalkU(\ve,\ve)$ is empty, then the first equality holds true.
Otherwise, take $\Walk=(\ve_0,\ldots,\ve_n) \in \mathcal{A}$ and note that $\Walk$ goes through $-\ve$. Let~$l$ be the smallest index such 
that $\ve_l=-\ve$, and let $k$ be the largest index smaller than $l$ such that $\ve_k=\ve$. We define a map $\Walk \mapsto \Walk'$ by
\[
\Walk'=(\ve_0,\ldots,\ve_{k-1},-\ve_l,-\ve_{l-1},\ldots,-\ve_k,\ve_{l+1},\ldots,\ve_n).
\]
It follows that $\Walk' \in \mathcal{A}$ and $(\Walk')'=\Walk$. 
By \eqref{eq:stepfact} and  \eqref{eq:inverted}, we see that $w(\Walk)=-w(\Walk')$, and therefore the sum of signed weights over $\mathcal{A}$ is zero.
To prove the second equality, observe that $\WalkU(\ve,\ve)$ maps bijectively to the space of finite sequences of walks from $\WalkV(\ve,\ve)$. Indeed, $(\ve)$ 
corresponds to the empty sequence of walks, and for $\Walk=(\ve_0,\ldots,\ve_n) \in \WalkU(\ve,\ve)$ of positive length, let $0=l_0<l_1<\ldots<l_m=n$ 
be the consecutive times when $\Walk$ visits~$\ve$, i.e.\ $\ve_{l_i} = \ve$ for 
$i \in \{0,\ldots,m\}$. Note that $\Walk_i = (\ve_{l_i},\ldots, \ve_{l_{i+1}}) \in \WalkV(\ve,\ve) $ for $i \in \{0,\ldots,m-1\}$.
It follows from~\eqref{eq:stepfact} that $w(\Walk) = \prod_{i=0}^{m-1} w(\Walk_{i})$. Hence, the sum of weights of all walks from $\WalkU(\ve,\ve)$,
which split into exactly $m$ walks from $\WalkV(\ve,\ve)$, equals the $m$th power of the sum of weights of all walks from $\WalkV(\ve,\ve)$.
Using the power series expansion of $(1-t)^{-1}$, we finish the proof.
\end{proof}

The second observation is that, when counting weights of walks going from $\ve$ to $\vg$, it is enough to look at these walks,
 which visit $\ue$ for the last time in the direction of $\ve$,
and afterwards visit $\ug$ for the first time in the direction of $\vg$.
\begin{lemma} \label{lem:walkfact}
For any $\ve,\vg \in \dG$ such that $\ue \neq \ug$,
\[
\sum_{\Walk \in \WalkW_{\FinG}(\ve,\vg)} w(\Walk)= \sum_{\Walk \in \WalkW_{\FinG}(\ve,\ve)} w(\Walk)  \sum_{\Walk \in \WalkV_{\FinG}(\ve,\vg)} 
w(\Walk)   \sum_{\Walk \in \WalkW_{\FinG \setminus \{\ue \}}(\vg,\vg)} w(\Walk).
\]
\end{lemma}
\begin{proof}
Again, if $\WalkW_{\FinG}(\ve,\vg)$ is empty, then $\WalkV_{\FinG}(\ve,\vg)$ is also empty and the statement is true.
Otherwise, for each $\Walk =(\ve_0,\ldots,\ve_n) \in \WalkW_{\FinG}(\ve,\vg)$, let $k$ be the largest index such that $\ue_k = \ue$, 
and let $l$ be the smallest index larger than $k$ such that $\ue_l =\ug$. We define $\Walk_{\ue\ue}= (\ve_0,\ldots,\ve_k)$, $\Walk_{\ue\ug} = (\ve_{k},\ldots,\ve_{l})$ and
$\Walk_{\ug\ug} = (\ve_l,\ldots,\ve_n)$. By \eqref{eq:stepfact}, we have that $w(\Walk)=w(\Walk_{\ue\ue})w(\Walk_{\ue\ug})w(\Walk_{\ug\ug})$. It follows from Lemma \ref{lem:zerosuminverse}
that the contribution of the walks $\Walk$ such that $\Walk_{\ue\ue} \in \WalkW_{\FinG}(\ve,-\ve)$ to the sum on the left-hand side of the desired equality is zero.
The same holds for the walks $\Walk$ with $\Walk_{\ug\ug} \in \WalkW_{\FinG \setminus  \{\ue \}}(-\vg,\vg)$.
Therefore, the only walks $\Walk$ that contribute to the sum satisfy $\Walk_{\ue \ue} \in \WalkW_{\FinG}(\ve,\ve)$, $\Walk_{\ue \ug} \in \WalkV_{\FinG}(\ve,\vg)$ 
and $\Walk_{\ug \ug} \in \WalkW_{\FinG \setminus  \{\ue \}}(\vg,\vg)$. 
\end{proof}

Note that $\WalkV_{\FinG}(\ve,\vg)$ may be empty even when $\WalkW_{\FinG}(\ve,\vg)$ is nonempty.

\subsubsection*{Dependence of Z on the graph}

The next result expresses a multiplicative relation between the generating functions 
of even subgraphs of $\FinG$ and $\FinG \setminus  \{\ue \}$ for some edge $\ue$. 
We will write $Z_{\FinG}$ to express the dependence of $Z$ on the graph~$\FinG$. 
\begin{corollary} \label{cor:Z} For any $\ve \in \dG$,
\begin{align*}
Z_{\FinG} = \Big(1-\sum_{\Walk \in \WalkV_{\FinG} (\ve,\ve) } w(\Walk) \Big) Z_{\FinG \setminus  \{\ue \} }.
\end{align*}
\end{corollary}
\begin{proof}
By \eqref{eq:defZ}, $Z$ is a sum of monomials in $x_{\ue}$, and therefore
\begin{align*}
Z_{\FinG}=Z_{\FinG \setminus \{\ue \} }+ x_{\ue}\frac{\partial}{\partial x_{\ue}} Z_{\FinG} \big|_{x_{\ue}=0}.
\end{align*}
To compute the partial derivative of $Z_{\FinG}$, we use the exponential formula from Theorem~\ref{thm:main}.
To justify why we obtain the sum over $\WalkV_{\FinG} (\ve,\ve)$, we make two observations: the only closed walks that survive the evaluation $x_{\ue}=0$ go through $\ue$ 
exactly once, and to each $\Walk \in \WalkV_{\FinG} (\ve,\ve) $ there correspond exactly $2|\Walk|$ closed walks with the same signed weight as $\Walk$, and which 
define the same, up to a time parametrization, closed curve in the plane.
Since putting $x_{\ue}=0$ is equivalent to removing $\ue$ from $\FinG$, we use Theorem~\ref{thm:main} again to express the exponential as $Z_{\FinG \setminus  \{\ue \}}$. 
Note that by~\eqref{eq:defZ} the partial derivative is actually constant in $x_{\ue}$. We still chose to evaluate it at zero since the fact that it does not depend on $x_{\ue}$ 
is not apparent when differentiating the exponential formula.
\end{proof}

\subsubsection*{Proof of Theorem \ref{thm:fermionwalks}}
\begin{proof}
The case $\vg = - \ve$ follows from Lemma \ref{lem:zerosuminverse} and the fact that $F_{\ve,-\ve}=0$. Next,
suppose that $ \vg = \ve $ and take $\FinH \in \fE (\ve,\ve)$. As mentioned before, $\FinH$ can be thought of as an even subgraph of $\FinG$ containing $\ue$.
It follows that the left-most path $\Walk_{\FinH}$ goes along the boundary of the (possibly unbounded) face of~$\FinH$ which lies on the 
left-hand side of $\ve$. It means that it does not have any self-crossings and therefore, by Theorem \ref{thm:Whitney}, $e^{\frac{i}{2} \alpha(\Walk_{\FinH})} = -1$. 
It follows from \eqref{eq:defZ} and \eqref{eq:deffermion} that
\begin{align*}
\overline{F_{\ve,\ve}} = 1 - \frac{1}{Z_{\FinG}}\mathop{\sum_{e \in \FinH \subset \FinG}}_{ \FinH \text{ even} } \prod_{\ug \in \FinH} x_{\ug} =  \frac{ Z_{\FinG \setminus \{ \ue \}}} {Z_{\FinG}}.
\end{align*}
Hence by Corollary \ref{cor:Z} and Lemma \ref{lem:sumsloops},
\begin{align} \label{eq:eecase}
\overline{F_{\ve,\ve}} =  \frac{ Z_{\FinG \setminus \{ \ue \}}} {Z_{\FinG}} = \Big(1-\sum_{\Walk \in \WalkV_{\FinG} (\ve,\ve) } w(\Walk) \Big)^{-1}= \sum_{\Walk \in \WalkW_{\FinG}(\ve,\ve)} w(\Walk).
\end{align}

The last case is when $\ue \neq \ug$. Let $\FinH \in \fE (\ve,\vg)$ and $\vgamma = (m(\ug), m(\ue))$. We put $x_{\ugamma}=1$. Without loss of generality we assume that $\FinH \cup \{\ugamma\}$ 
satisfies the definition of a graph, i.e.\ no vertices of $\FinH$ lie on $\gamma$. Indeed, if this is not the case, then we can add two edges 
$\ugamma_1 = \{m(\ue),v\}$ and $\ugamma_2=\{v, m(\ug)\}$, for some suitably chosen vertex $v$ (see Figure \ref{fig1}). The rest of the proof can 
be easily adjusted to this situation. Note that $\FinH \cup \{\ugamma\}$ is an even subgraph of~$\FinG_{\ve,\vg} \cup \{\ugamma\}$.

Let $\Walk^{\circ}_{\FinH}$ be the closed path that starts at $\vgamma$ and then agrees with $\Walk_{\FinH}$ until it goes back to $\vgamma$. 
We claim that
\begin{align} \label{eq:equalsign}
(-1)^{C(\FinH \cup \{\ugamma\})}= (-1)^{C(\Walk^{\circ}_{\FinH})}= -e^{\frac{i}{2}\alpha(\Walk^{\circ}_{\FinH})}= -e^{\frac{i}{2}(\alpha(\Walk_{\FinH})+\beta)},
\end{align}
where $\beta= \angle(\vg,\vgamma) + \angle(\vgamma,\ve) $. The second equality is a consequence of Theorem~\ref{thm:Whitney}, and the last 
one follows directly from the definitions \eqref{eq:defangle} and \eqref{eq:defweight}.
Since $\FinH$ is embedded in the plane without edge-crossings, ${C(\FinH \cup \{\ugamma\})}$ is the number of edges in $\FinH$ which are 
crossed by $\gamma$. Similarly, since $\Walk_{\FinH}$ always makes a step to the left-most edge, $\Walk^{\circ}_{\FinH}$ does not have any 
self-crossings at the vertices of $\FinH \cup \{\ugamma\}$. Therefore, ${C(\Walk^{\circ}_{\FinH})}$ is equal to the number of edges 
in~$\ES(\Walk^{\circ}_{\FinH})$ which cross $\gamma$. 
What is left to prove, is that the number of edges in~$\FinH \setminus \ES(\Walk^{\circ}_{\FinH})$ which are crossed by $\gamma$ is even.
To this end, let $\{\Walk_1, \ldots, \Walk_k \}$ be a collection of edge-disjoint closed paths, such that 
$\FinH \setminus \ES(\Walk^{\circ}_{\FinH}) = \bigcup_{i=1}^k \ES(\Walk_i)$. Again, since $\Walk_{\FinH}$ is the left-most path in $\FinH$,
 it is true that $\Walk_{\FinH}^{\circ}$ does not have any crossings with $\Walk_i$, $i=1,\ldots,k$, at the
vertices of $\FinH \cup \{\ugamma\}$. It follows that $C(\Walk_{\FinH}^{\circ}, \Walk_i)$ is the number of edges in $\ES(\Walk_i)$ crossed by $\gamma$.
Since $C(\Walk, \Walk')$ is even for any two closed paths $\Walk$ and $\Walk'$, we have established \eqref{eq:equalsign}.

Note that $\FinH \mapsto \FinH \cup \{\ugamma\}$ is a bijection between  $\fE (\ve,\vg)$ and the collection of even subgraphs of $\FinG_{\ve,\vg} \cup \{\ugamma\}$ which contain $\gamma$.
Similarly to the previous case, from \eqref{eq:equalsign}, \eqref{eq:defZ} and \eqref{eq:deffermion}, it follows that
\begin{align} \label{eq:FZ}
 e^{\frac{i}{2}\beta}  \overline{F_{\ve,\vg}} &=\frac{Z_{\FinG_{\ve,\vg}}-Z_{\FinG_{\ve,\vg} \cup \{\ugamma\}} }{Z_{\FinG}} = \Big( 1- \frac{Z_{\FinG_{\ve,\vg}\cup \{\ugamma\}}}{ Z_{\FinG_{\ve,\vg}} }\Big) 
\frac{Z_{\FinG \setminus \{ \ue, \ug\} }}{Z_{\FinG \setminus \{ e \}}} \frac{Z_{\FinG \setminus \{\ue\} }}{Z_{\FinG}},
\end{align}
where we also used the fact that $Z_{\FinG_{\ve,\vg}} = Z_{\FinG \setminus \{ \ue, \ug\}}$.
Using Corollary \ref{cor:Z}, we get
\begin{align*}
1- \frac{Z_{\FinG_{\ve,\vg} \cup \{\ugamma\}}}{ Z_{\FinG_{\ve,\vg}} } = \sum_{ \Walk \in \WalkV_{\FinG_{\ve,\vg} \cup \{\ugamma\}} (\vgamma,\vgamma) } w(\Walk)  
= e^{\frac{i}{2}\beta}   \sum_{ \Walk \in \WalkV_{\FinG} (\ve,\vg) } w(\Walk).
\end{align*}
Just as in \eqref{eq:eecase}, the remaining ratios of generating functions in \eqref{eq:FZ} can be expressed in terms of walks, and therefore
\begin{align*}
\overline{F_{\ve,\vg}} & = \sum_{ \Walk \in \WalkV_{\FinG} (\ve,\vg) } w(\Walk)\sum_{\Walk \in \WalkW_{\FinG \setminus \{\ue\} }(\vg,\vg)} w(\Walk) 
  \sum_{\Walk \in \WalkW_{\FinG}(\ve,\ve)} w(\Walk) \\
	&= \sum_{\Walk \in \WalkW_{\FinG}(\ve,\vg)} w(\Walk).
\end{align*}
The last equality follows from Lemma \ref{lem:walkfact}.
\end{proof}

\subsection{Proof of Theorem \ref{thm:KWshol}}
\label{sec:proofKWshol}
\begin{proof}
We put $\Sf =\CI f$. Take two consecutive edges~$\ve_1$ and $\ve_2$ from $\In(z)$ ordered counterclockwise around~$z$ and suppose that 
\begin{align} \label{eq:shol0} 
x_{\ue_1}^{-1} \KW  \Sf(\ve_1) =x_{\ue_2}^{-1} e^{\frac{i}{2} \angle(\ve_1,\ve_2)} \KW \Sf (\ve_2) .
\end{align}
By the definition of $\KW $, this is equivalent to
\begin{align*}
&\Sf(\ve_1)x_{\ue_1}^{-1}-e^{\frac{i}{2}\angle(\ve_1,-\ve_2)}\Sf(-\ve_2) -  \sum_{\vg \in \Out(z) \setminus  \{-\ve_1,-\ve_2\}} \mkern-18mu \Sf(\vg) e^{\frac{i}{2} \angle(\ve_1,\vg)} = \\
&e^{\frac{i}{2}\angle(\ve_1,\ve_2)}\Big(\Sf(\ve_2)x_{\ue_2}^{-1}-e^{\frac{i}{2}\angle(\ve_2,-\ve_1)}\Sf(-\ve_1) - 
\sum_{\vg \in \Out(z) \setminus  \{-\ve_1,-\ve_2\}} \mkern-18mu \Sf(\vg) e^{\frac{i}{2} \angle(\ve_2,\vg)}\Big).
\end{align*}
Since the faces of $\InfG$ are convex, $\angle(\ve_1,\ve_2) = \theta_{\ue_1}+\theta_{\ue_2}>0$.
Using basic properties of the complex argument, one obtains that 
$\angle(\ve_1, \ve_2)+ \angle(\ve_2,-\ve_1)= \pi$, and
\begin{align*} 
\angle(\ve_1, \ve_2)+\angle(\ve_2, \vg) = \angle(\ve_1, \vg) \quad \text{for all} \ \vg \in \Out(z) \setminus \{-\ve_1, -\ve_2 \}.
\end{align*}
Combining this with the equation above, gives
\begin{align} \label{eq:shol1}
\Sf(\ve_1)x_{\ue_1}^{-1} + i \Sf(-\ve_1) = e^{\frac{i}{2}\angle(\ve_1,\ve_2)} (\Sf(\ve_2)x_{\ue_2}^{-1} - i\Sf(-\ve_2)),
\end{align}
which using criticality of the weights, yields
\begin{align} \label{eq:shol2}
 &\big(\sin(\theta_{\ue_1}/2)\big)^{-1} e^{-\frac{i}{2} \theta_{\ue_1}} \Big(\Sf(\ve_1) \cos(\theta_{\ue_1}/2) + i \Sf(-\ve_1)\sin(\theta_{\ue_1}/2) \Big) = \\  \nonumber
 &\big(\sin (\theta_{\ue_2}/2)\big)^{-1} e^{\frac{i}{2} \theta_{\ue_2}} \Big( \Sf(\ve_2)\cos( \theta_{\ue_2}/2) - i \Sf(-\ve_2)\sin(\theta_{\ue_2}/2)\Big).
\end{align}
We put 
\[
\ell= e^{-\frac{i} {2}\theta_{\ue_1}}\ell_{\ve_1}= e^{\frac{i} {2} \theta_{\ue_2}} \ell_{\ve_2}  = (z -z^*)^{-\frac{1}{2}} \R,
\] 
where $z^*$ is the dual vertex corresponding to the face lying on the right-hand side of $\ve_1$ and $-\ve_2$.
From basic geometry it follows that 
\begin{align*}
\Proj( z; z e^{i \beta} \R) = z e^{i\beta} \cos \beta \qquad \text{and} \qquad \Proj( z; z i e^{i \beta} \R) = -zi e^{i\beta} \sin \beta 
\end{align*}
for any nonzero complex number $z$ and any real number $\beta$.
This, together with the definition of $S$, the fact that $\Sf(\ve) \in \ell_{\ve}$ and $\ell_{\ve} = i\ell_{-\ve}$, implies that equation \eqref{eq:shol2} takes the form 
\begin{align} \label{eq:shol3}
 \Proj( f(\ue_1) ;  \ell) =  \Proj( f(\ue_2);\ell).
\end{align}
Therefore, condition \eqref{eq:shol0} is equivalent to condition~\eqref{eq:shol3}. 

Now, assume that~\eqref{eq:shol0} holds for all pairs of consecutive edges in 
$\In(z)=\{\ve_1,\ve_2, \ldots, \ve_k\}$. We obtain that 
$\KW \Sf(\ve_1) = e^{\frac{i}{2}\sum_{i=1}^k \angle(\ve_{i}, \ve_{i+1})} \KW \Sf(\ve_1) = - \KW \Sf(\ve_1)$, where $\ve_{k+1}=\ve_1$, and 
hence, $\KW \Sf(\ve)=0$ for all $\ve \in \In(z)$. The opposite implication uses the fact that condition \eqref{eq:shol0} 
is trivially satisfied when $T \Sf(\ve_1) = T \Sf(\ve_2) =0 $.
\end{proof}

\subsection{Proof of Theorem \ref{thm:walkrepresentation}}
\label{sec:proowalkrepresentation}
\subsubsection*{Matrices of operators} \label{sec:Hilbertspace}
If $\FinG$ is finite, then $\ell^2=\ell^2(\dG)$ is a finite dimensional Euclidean space and hence all
automorphisms of $\ell^2$ are continuous and are expressed via matrix multiplication. 
If $\FinG$ is infinite, then $\ell^2$ is an infinite dimensional Hilbert space and
all continuous automorphisms of $\ell^2$ are also given by (infinite) matrix multiplication. To be precise, let $\{i_{\ve} \}_{\ve \in \dG}$, 
where $i_{\ve}(\vg)= \delta_{\ve,\vg}$,
be the standard basis of $\ell^2$ and let $\langle \cdot,\cdot \rangle$ be the inner product in~$\ell^2$. 
If $A$ is a continuous automorphism of~$\ell^2$, then $A_{\ve,\vg} := \langle A i_{\vg},  i_{\ve} \rangle $ are the entries of the associated matrix, 
and $A$ acts via matrix multiplication, i.e.\
\begin{align*}
A \varphi (\ve) = \sum_{\vg \in \dG} A_{\ve,\vg} \varphi(\vg) \qquad \text{for all} \ \varphi \in \ell^2.
\end{align*}
The rows (and also columns) of $A$ belong to $\ell^2$ and hence the order of summation is irrelevant. 
Moreover, the matrix of a composition of two bounded operators is the product of the two matrices of these operators.
Note that the entries of the matrix are 
given in terms of linear functionals. Hence, whenever a sequence of operators converges in the weak topology, 
the entries also converge to the entries of the matrix of the limiting operator.

By $\|A\|$ and $\rho(A)$ we will denote the operator norm and the spectral radius of $A$.
From the theory of Banach algebras, we know that
\begin{align} \label{eq:powerseriesA}
\rho(A) = \lim_{n \rightarrow \infty } \| A^n \| ^{1/n}, \quad \text{and hence} \quad (\Id-A)^{-1} = \sum_{n=0}^{\infty} A^n
\end{align}
if $\rho(A)<1$.
Here, $\Id$ is the identity on $\ell^2$ and the limit is taken in the operator norm.

\subsubsection*{Supercritical case} 
\begin{proof}[Proof of exponential decay]
Let $\FinG$ be any subgraph of $\InfG$ and fix $\beta \in (0,1)$. 
Let the weight system $x(\beta)$ and the coupling constants $J$ be as in \eqref{eq:isingweights}.
By monotonicity of the hyperbolic tangent and by compactness,
\begin{align} \label{eq:epsilonbound}
 \epsilon:=\sup_{\ue \in \FinG} \frac{x_{\ue}(\beta)}{x_{\ue}(1)} = 
\sup_{\ue \in \FinG} \frac{\tanh\beta J_{\ue}}{\tanh J_{\ue}} 
\leq \sup_{j \in [m,M]} \frac{\tanh\beta j}{\tanh j} <1,
\end{align}
where $\tanh m= \tan(k/2)$ and $\tanh M = \tan(K/2)$, with $k$ and $K$ as in~\eqref{eq:isoradialcond}.

Let $D$ be an isomorphism of $\ell^2$, which for each directed edge $\ve$ rescales the coordinate corresponding to $\ve$ by $\sqrt{x_{\ue}(\beta)}$.
Because of condition $\eqref{eq:isoradialcond}$, $D$ is bounded and has a bounded inverse $D^{-1}$.
Let $B =D^{-1} \Lambda D$, where $\Lambda$ is the transition matrix for $\FinG$ and the weight system $x(\beta)$. 
In the language of~\cite{Lis}, $B$ is a conjugated Kac-Ward transition matrix. We will use Corollary 3.2 from \cite{Lis}, which explicitly gives
the operator norm of a conjugated transition matrix.
To this end, note that the angles $\theta$ sum up to $\pi$ around each vertex of $\InfG$ (see Figure~\ref{fig2}).
Hence, from \eqref{eq:isingweights} and \eqref{eq:epsilonbound} it follows that 
\begin{align} 
\sum_{\ve \in \text{\textnormal{Out}}(z)}\arctan \big(x_{\ue}(\beta)/\epsilon\big)& \leq 
\sum_{\ve \in \text{\textnormal{Out}}(z)}\arctan \big(x_{\ue}(1)\big)  \\
&= \sum_{\ve \in \text{\textnormal{Out}}(z)} \theta_{\ue}/2  \leq \pi /2 \nonumber
\end{align}
for all vertices $z$. 
From the above inequality, Corollary 3.2 and Remark 1 from \cite{Lis}, it follows that the operator norm of $B$ is bounded from above
by $\epsilon$. The operator norm gives an upper bound on the spectral radius, and since $B$ has the same spectrum 
as $\Lambda$, the spectral radius of $\Lambda$ is not larger than $\epsilon$.
To compute the inverse Kac-Ward operator, we can therefore use the power series expansion \eqref{eq:powerseriesA} with $A = \Lambda$.
To get the non-backtracking walk representation, we compute the powers of $\Lambda$ using matrix multiplication and we use identity~\eqref{eq:stepfact}.
We also use the fact that convergence in norm is stronger than weak convergence and hence implies convergence of the entries of the 
corresponding matrices.

Furthermore, note that for all $\ve$ and~$\vg$
\begin{align*}
|\Lambda^n_{\ve,\vg}| &= |(D B^n D^{-1})_{\ve,\vg}| = |\langle D B^n D^{-1} i_{\vg},i_{\ve}\rangle| 
					\leq \| D B^n D^{-1} i_{\vg} \| \cdot \|i_{\ve}\| \\
					 & \leq \| D B^n D^{-1}  \| \leq  \|D\| \cdot \|D^{-1}\| \cdot \|B\|^n \leq C\epsilon^n,
\end{align*}
where we used the Cauchy-Schwarz inequality and submultiplicativity of the operator norm. 
Note that both $C$ and $\epsilon$ are universal for all subgraphs~$\FinG$, and moreover, \eqref{eq:epsilonbound} and \eqref{eq:isoradialcond} 
give explicit upper bounds on these constants. 
\end{proof}

\subsubsection*{Critical case} 
\begin{proof}[Proof of the non-backtracking walk expansion]
As we already mentioned, the power expansion formula \eqref{eq:powerseriesA} is valid 
whenever the spectral radius of the transition matrix is strictly smaller than one. 
We will now prove that this is the case if $\FinG$ is a finite subgraph of $\InfG$ and the weight system is critical. 

Our main tool will be Corollary 3.3 from \cite{Lis}, and therefore we need the following definition:
if $\dx$ is a system of weights on the directed edges of~$\FinG$, then we define $\xi^z(\dx)$ to be the
unique positive solution in $s$ of the equation 
\begin{align} \label{eq:solutiondef}
\sum_{\ve \in \text{\textnormal{Out}}(z)}\arctan \big(|\dx_{\ve}|^2/s\big)= \pi /2.
\end{align}
By Corollary 3.3 from \cite{Lis}, for the spectral radius of the critical transition matrix to be strictly smaller than one, 
it is enough to construct a weight system $\dx$ such that $\xi_z(\dx)<1$ for all vertices $z$, and
\begin{align} \label{eq:factorizes}
\dx_{\ve} \dx_{-\ve} =x_{\ue}(1)= \tan(\theta_{\ue}/2)\qquad \text{for all}\ \ve \in \dG. 
\end{align}

Let $V$ be the set of vertices of $\FinG$ and let $\partial \FinG \subset V$ be the set of vertices, whose degree in $\FinG$ is smaller than in $\InfG$.
Moreover, let $\partial_r\FinG \subset V$ denote the set of vertices whose graph distance to $\partial \FinG$ is at most $r$. 
We will inductively construct weight systems $\dx^r$, which satisfy \eqref{eq:factorizes}, and for which 
\begin{align} \label{eq:induction}
k_z(\dx^r) < 1\  \text{for all} \ z \in \partial_r\FinG, \quad \text{and} \quad k_z(\dx^r) = 1 \ \text{for all} \ z \in V \setminus \partial_r\FinG.
\end{align} 
Indeed, let $\dx_{\ve}^0= \sqrt{\tan(\theta_{\ue}/2)}$. The angles $\theta$ sum up to $\pi$ around each vertex in $\InfG$ and hence
$\xi_z(\dx^0) = 1$ for all $z \in V\setminus \partial \FinG$.
Since removing an edge incident to a vertex $z$ strictly decreases~$\xi_z$, $\xi_z(\dx^0) < 1$ for all $z \in \partial \FinG$.
Therefore, $\dx^0$ gives the basis of our induction. Now, assume that we already constructed an~$\dx^r$ which satisfies \eqref{eq:factorizes} and \eqref{eq:induction}.
If $\partial_r\FinG=V$, then $\dx = \dx^r$ yields the desired bound on the spectral radius. Otherwise, take any $z \in \partial_{r+1}\FinG \setminus \partial_r \FinG$ and any $w \in \partial_r \FinG$ 
at distance one from~$z$, i.e. such that $\ve=(w,z)\in \dG$. By the induction hypothesis, $\xi_z(\dx^r)=1$ and $\xi_w(\dx^r) < 1$.
By continuity, one can slightly increase $\dx^r_{\ve}$ so that still $\xi_w < 1$. To still satisfy \eqref{eq:factorizes}, 
the product over the two opposite orientations of $\ue$ has to remain constant,
and hence one has to slightly decrease $\dx^r_{-\ve}$ which results in $\xi_z <1$. The value of $\xi$ at other vertices does not change.
If we do this procedure for all $z \in \partial_{r+1}\FinG$, it means that we constructed $\dx^{r+1}$ which satisfies \eqref{eq:factorizes} and \eqref{eq:induction}.
We proceed until we cover all vertices of $\FinG$. Note that finiteness of $\FinG$ is crucial in this reasoning.
\end{proof}

\begin{proof}[The full Kac-Ward operator does not have a bounded inverse] Consider the \\* critical Kac-Ward operator on the full graph $\InfG$.
We already proved that it is not invertible when treated as an operator on the vector space $\C^{\vec{\InfG}}$ since constant
functions are in the kernel of $TS$, where $S$ is the projection operator from Section \ref{sec:isoradial}.

The idea is similar when $T$ is seen as a continuous operator on $\ell^2$. We will consider
elements of $\ell^2$ which ``approximate'' constant functions and show that their images under $TS$ are close to zero. 
To this end, let $f_{\FinH} \in \C^{\InfG}$ be the characteristic function of a finite graph $\FinH \subset \InfG$. 
By the definition of~$S$, $\varphi_{\FinH}:=S f_{\FinH} \in \ell^2$ and $\| \varphi_{\FinH} \| \geq \sin(k/2) \sqrt{|\FinH|}$, 
where $k$ is as in~\eqref{eq:isoradialcond}.
Note that~$f_{\FinH}$ is s-holomorphic at all interior vertices of $\FinH$ and $\InfG \setminus \FinH$. 
By Theorem~\ref{thm:KWshol}, $T \varphi_{\FinH}$ can be nonzero only at these directed edges, which point at the vertices of $\partial \FinH$, where $\partial \FinH$
is as in the previous proof.
From the definition of the Kac-Ward operator, it follows that 
\begin{align*}
\| T \varphi_{\FinH} \|_{\infty} \leq \Delta \| x(1) \|_{\infty} \|\varphi_{\FinH}\|_{\infty} \leq \tan (K/2) \Delta ,
\end{align*}
where $\Delta$ is the maximal degree of $\InfG$, and $K$ is as in~\eqref{eq:isoradialcond}. Hence, 
\begin{align*}
\| T \varphi_{\FinH} \| \leq \tan (K/2)  \Delta^{3/2} \sqrt{|\partial \FinH|},
\end{align*} 
which in the end yields 
\begin{align*}
\| T \varphi_{\FinH}\| \leq C \sqrt{|\partial \FinH| / | \FinH|} \|\varphi_{\FinH}\|,
\end{align*}
for some constant $C$ independent of $\FinH$. 

It is now enough to notice that $\InfG$ admits subgraphs for which the ratio $|\partial \FinH| / | \FinH|$ is arbitrarily small; 
it will mean that the inverse, if it exists, is unbounded in norm.
To this end, one can consider subgraphs $\FinH^r$, which are induced by the vertices of $\InfG$ contained in the square $[-r,r]\times[-r,r]$.
Using condition~\eqref{eq:isoradialcond}, which says that all edges of $\InfG$ are surrounded by disjoint rhombi of positive minimal area (and also
finite maximal area), one can prove that $|\partial \FinH^r|$ and $|\FinH^r|$ grow like $r$ and $r^2$ respectively when $r$ goes to infinity. 
\end{proof}

\textbf{Acknowledgments.} The author would like to thank Wouter Kager and Ronald Meester for introducing him to the combinatorial aspects of the Ising model
 and for useful remarks on the manuscript, and Federico Camia for the encouragement to write this paper. 
The research was supported by NWO grant Vidi 639.032.916.

\bibliographystyle{amsplain}
\bibliography{Ising}

\begin{bibdiv}
\begin{biblist}

\bib{Baxter}{article}{
      author={Baxter, Rodney~James},
       title={Free-fermion, checkerboard and {${\bf Z}$}-invariant lattice
  models in statistical mechanics},
        date={1986},
        ISSN={0962-8444},
     journal={Proc. Roy. Soc. London Ser. A},
      volume={404},
      number={1826},
       pages={1\ndash 33},
      review={\MR{836281 (87h:82078)}},
}

\bib{BefDum}{article}{
      author={Beffara, Vincent},
      author={Duminil-Copin, Hugo},
       title={Smirnov's fermionic observable away from criticality},
        date={2012},
     journal={Ann. Probab.},
      volume={40},
      number={6},
       pages={2667\ndash 2689},
}

\bib{BoutTil1}{article}{
      author={Boutillier, C{\'e}dric},
      author={de~Tili{\`e}re, B{\'e}atrice},
       title={The critical {${\bf Z}$}-invariant {I}sing model via dimers: the
  periodic case},
        date={2010},
        ISSN={0178-8051},
     journal={Probab. Theory Related Fields},
      volume={147},
      number={3-4},
       pages={379\ndash 413},
         url={http://dx.doi.org/10.1007/s00440-009-0210-1},
      review={\MR{2639710 (2012a:82010)}},
}

\bib{BoutTil2}{article}{
      author={Boutillier, C{\'e}dric},
      author={de~Tili{\`e}re, B{\'e}atrice},
       title={The critical {$Z$}-invariant {I}sing model via dimers: locality
  property},
        date={2011},
        ISSN={0010-3616},
     journal={Comm. Math. Phys.},
      volume={301},
      number={2},
       pages={473\ndash 516},
         url={http://dx.doi.org/10.1007/s00220-010-1151-3},
      review={\MR{2764995 (2011m:82009)}},
}

\bib{BFS}{article}{
      author={Brydges, David},
      author={Fr{\"o}hlich, J{\"u}rg},
      author={Spencer, Thomas},
       title={The random walk representation of classical spin systems and
  correlation inequalities},
        date={1982},
        ISSN={0010-3616},
     journal={Comm. Math. Phys.},
      volume={83},
      number={1},
       pages={123\ndash 150},
         url={http://projecteuclid.org/getRecord?id=euclid.cmp/1103920749},
      review={\MR{648362 (83i:82032)}},
}

\bib{CHI}{article}{
      author={Chelkak, Dmitry},
      author={Hongler, Cl{\'e}ment},
      author={Izyurov, Konstantin},
       title={Conformal invariance of spin correlations in the planar {I}sing
  model},
        date={2015},
        ISSN={0003-486X},
     journal={Ann. of Math. (2)},
      volume={181},
      number={3},
       pages={1087\ndash 1138},
         url={http://dx.doi.org/10.4007/annals.2015.181.3.5},
      review={\MR{3296821}},
}

\bib{CI}{article}{
      author={Chelkak, Dmitry},
      author={Izyurov, Konstantin},
       title={Holomorphic spinor observables in the critical {I}sing model},
        date={2013},
        ISSN={0010-3616},
     journal={Comm. Math. Phys.},
      volume={322},
      number={2},
       pages={303\ndash 332},
         url={http://dx.doi.org/10.1007/s00220-013-1763-5},
      review={\MR{3077917}},
}

\bib{ChelkSmir}{article}{
      author={Chelkak, Dmitry},
      author={Smirnov, Stanislav},
       title={Universality in the 2{D} {I}sing model and conformal invariance
  of fermionic observables},
        date={2012},
        ISSN={0020-9910},
     journal={Invent. Math.},
      volume={189},
      number={3},
       pages={515\ndash 580},
         url={http://dx.doi.org/10.1007/s00222-011-0371-2},
      review={\MR{2957303}},
}

\bib{Cimasoni}{article}{
      author={Cimasoni, David},
       title={A generalized {K}ac-{W}ard formula},
        date={2010JUL},
        ISSN={1742-5468},
     journal={J. Stat. Mech. Theory E.},
       pages={P07023},
}

\bib{CimDum}{article}{
      author={Cimasoni, David},
      author={Duminil-Copin, Hugo},
       title={The critical temperature for the {I}sing model on planar doubly
  periodic graphs},
        date={2013},
        ISSN={1083-6489},
     journal={Electron. J. Probab.},
      volume={18},
       pages={no. 44, 1\ndash 18},
         url={http://ejp.ejpecp.org/article/view/2352},
}

\bib{Dubedat}{unpublished}{
      author={Dub{\'e}dat, Julien},
       title={{Exact bosonization of the Ising model}},
        date={2011},
        note={arXiv:1112.4399},
}

\bib{Helm}{article}{
      author={Helmuth, Tyler},
       title={{Ising model observables and non-backtracking walks}},
        date={2014},
     journal={Journal of Mathematical Physics},
      volume={55},
      number={8},
  url={http://scitation.aip.org/content/aip/journal/jmp/55/8/10.1063/1.4881723},
}

\bib{HKZ}{unpublished}{
      author={Hongler, Cl{\'e}ment},
      author={Kyt\"{o}l\"{a}, Kalle},
      author={Zahabi, Ali},
       title={Discrete holomorphicity and {I}sing model operator formalism},
        date={2012},
        note={arXiv:1211.7299},
}

\bib{HongSmir}{article}{
      author={Hongler, Cl{\'e}ment},
      author={Smirnov, Stanislav},
       title={The energy density in the planar {I}sing model},
        date={2013},
        ISSN={0001-5962},
     journal={Acta Math.},
      volume={211},
      number={2},
       pages={191\ndash 225},
         url={http://dx.doi.org/10.1007/s11511-013-0102-1},
      review={\MR{3143889}},
}

\bib{KacWard}{article}{
      author={Kac, Mark},
      author={Ward, John~Clive},
       title={A combinatorial solution of the two-dimensional {I}sing model},
        date={1952},
        ISSN={0031-899X},
     journal={Phys. Rev.},
      volume={88},
      number={6},
       pages={1332\ndash 1337},
}

\bib{KLM}{article}{
      author={Kager, Wouter},
      author={Lis, Marcin},
      author={Meester, Ronald},
       title={The signed loop approach to the {I}sing model: foundations and
  critical point},
        date={2013},
        ISSN={0022-4715},
     journal={J. Stat. Phys.},
      volume={152},
      number={2},
       pages={353\ndash 387},
         url={http://dx.doi.org/10.1007/s10955-013-0767-z},
      review={\MR{3082653}},
}

\bib{Kasteleyn}{article}{
      author={Kasteleyn, Piet},
       title={{The statistics of dimers on a lattice. I. The number of dimer
  arrangements on a quadratic lattice }},
        date={1961},
     journal={Physica},
      number={27},
       pages={1209\ndash 1225},
}

\bib{Kaufman}{article}{
      author={Kaufman, Bruria},
       title={Crystal statistics. {II}. {P}artition function evaluated by
  spinor analysis},
        date={1949},
        ISSN={0031-899X},
     journal={Phys. Rev.},
      volume={76},
      number={8},
       pages={1232\ndash 1243},
}

\bib{Ken1}{article}{
      author={Kenyon, Richard},
       title={Conformal invariance of domino tiling},
        date={2000},
        ISSN={0091-1798},
     journal={Ann. Probab.},
      volume={28},
      number={2},
       pages={759\ndash 795},
         url={http://dx.doi.org/10.1214/aop/1019160260},
      review={\MR{1782431 (2002e:52022)}},
}

\bib{Lis}{article}{
      author={Lis, Marcin},
       title={Phase transition free regions in the {I}sing model via the
  {K}ac-{W}ard operator},
        date={2014},
        ISSN={0010-3616},
     journal={Comm. Math. Phys.},
      volume={331},
      number={3},
       pages={1071\ndash 1086},
         url={http://dx.doi.org/10.1007/s00220-014-2061-6},
      review={\MR{3248059}},
}

\bib{Onsager}{article}{
      author={Onsager, Lars},
       title={Crystal statistics. {I}. {A} two-dimensional model with an
  order-disorder transition},
        date={1944},
     journal={Phys. Rev. (2)},
      volume={65},
       pages={117\ndash 149},
      review={\MR{0010315 (5,280d)}},
}

\bib{Smir2}{incollection}{
      author={Smirnov, Stanislav},
       title={Towards conformal invariance of 2{D} lattice models},
        date={2006},
   booktitle={International {C}ongress of {M}athematicians. {V}ol. {II}},
   publisher={Eur. Math. Soc., Z\"urich},
       pages={1421\ndash 1451},
      review={\MR{2275653 (2008g:82026)}},
}

\bib{Smir1}{article}{
      author={Smirnov, Stanislav},
       title={Conformal invariance in random cluster models. {I}. {H}olomorphic
  fermions in the {I}sing model},
        date={2010},
        ISSN={0003-486X},
     journal={Ann. of Math. (2)},
      volume={172},
      number={2},
       pages={1435\ndash 1467},
         url={http://dx.doi.org/10.4007/annals.2010.172.1441},
      review={\MR{2680496 (2011m:60302)}},
}

\bib{Whitney}{article}{
      author={Whitney, Hassler},
       title={On regular closed curves in the plane},
        date={1937},
        ISSN={0010-437X},
     journal={Compositio Math.},
      volume={4},
       pages={276\ndash 284},
         url={http://www.numdam.org/item?id=CM_1937__4__276_0},
      review={\MR{1556973}},
}

\end{biblist}
\end{bibdiv}

\end{document}